\documentclass[journal,a4paper,twocolumn]{revtex4}

\usepackage[T1]{fontenc}
\usepackage[utf8]{inputenc}

\usepackage[cmex10]{amsmath}
\usepackage{graphics}
\usepackage{dsfont}
\usepackage{bbm}
\usepackage{amssymb}
\usepackage{graphicx}
\usepackage{mathrsfs}

\usepackage[colorlinks=true,linkcolor=black,citecolor=black,plainpages=false,pdfpagelabels]{hyperref}

\newtheorem{theorem}{Theorem}

\newtheorem{proposition}[theorem]{Proposition}
\newtheorem{corollary}[theorem]{Corollary}
\newenvironment{proof}[1][Proof]{\begin{trivlist}
\item[\hskip \labelsep {\bfseries #1}]}{\end{trivlist}}

\newenvironment{remark}[1][Remark]{\begin{trivlist}
\item[\hskip \labelsep {\bfseries #1}]}{\end{trivlist}}

\newcommand{\abs}[1]{\ensuremath{|#1|}}

\newcommand{\norm}[2]{\ensuremath{|\!|#1|\!|_{#2}}}

\newcommand{\Norm}[2]{\ensuremath{\left|\!\left|#1\right|\!\right|_{#2}}}

\newcommand{\tr}{\textnormal{tr}}
\newcommand{\trace}[1]{\ensuremath{\tr (#1)}}

\newcommand{\ket}[1]{| #1 \rangle}

\newcommand{\bracket}[3]{\langle #1 | #2 | #3 \rangle}

\newcommand{\proj}[2]{| #1 \rangle\!\langle #2 |}

\renewcommand{\d}[1]{\ensuremath{\textnormal{d}#1}}

\newcommand{\cE}{\mathcal{E}}

\newcommand{\cL}{\mathcal{L}}
\newcommand{\cM}{\mathcal{M}}

\newcommand{\cP}{\mathcal{P}}

\newcommand{\cT}{\mathcal{T}}

\newcommand{\cZ}{\mathcal{Z}}

\newcommand{\id}{{\rm{id}}}

\def\>{{\rangle}}
\def\<{{\langle}}
\newcommand{\be}{\begin{equation}}
\newcommand{\ee}{\end{equation}}
\newcommand{\bea}{\begin{eqnarray}}
\newcommand{\eea}{\end{eqnarray}}
\newcommand{\C}{\mathbb{C}}
\newcommand{\1}{\mathbbm{1}}
\newcommand{\ra}{\rightarrow}

\newcommand{\qed}{\hfill\fbox\\\medskip }

\begin{document}
\title{\Large\sc Perturbation Bounds for Quantum Markov Processes\\ and their Fixed Points}
\author{Oleg Szehr, Michael~M.~Wolf}
\affiliation{\vspace*{3pt}Department of Mathematics, Technische Universit\"at M\"unchen, 85748 Garching, Germany\vspace*{3pt}}
\date{\today}


\begin{abstract}  
We investigate the stability of quantum Markov processes with respect to perturbations of their transition maps. In the first part, we introduce a condition number that measures the sensitivity of fixed points of a quantum channel to perturbations. We establish upper and lower bounds on this condition number in terms of subdominant eigenvalues of the transition map. 

In the second part, we consider quantum Markov processes that converge to a unique stationary state and we analyze the stability of the evolution at finite times. In this way we obtain a linear relation between the mixing time of a quantum Markov process and the sensitivity of its fixed point with respect to perturbations of the transition map.

\end{abstract}
\maketitle
\section{Introduction} Quantum Markov Processes naturally occur in various directions of quantum physics such as quantum statistical physics, quantum optics or quantum information theory. Whenever the time evolution of some quantum system does not depend on its history, it can be appropriately described as a quantum Markov process.  Here we have in particular in mind evolutions of open quantum systems which eventually converge to a set of stationary states.

Such evolutions either arise naturally from relaxation or equilibration, or they may be engineered for instance for the purpose of dissipative quantum computation \cite{VWC2008}, dissipative quantum state preparation \cite{VWC2008, DiehlNature, KrausPRA}  or quantum Metropolis sampling \cite{QMetropolis}. In those cases, the quantum Markov chain is designed so that it drives any initial state towards a sought target---preferably as rapid as possible.

The present work is devoted to the question how sensitive stationary states are to perturbations of the transition map of the corresponding Markov chain. While this appears to be a well studied subject for classical Markov chains \cite{Overview,Sen1,Sen2,Mit1,Mit2}, it is, to the best of our knowledge, essentially untouched territory for their quantum counterparts. Guided by the classical theory we will follow two alternative approaches, both of which result in an inequality of the form 
\be\|\rho_1-\rho_2\|\leq\kappa\|\cT_1-\cT_2\|.\label{eq:targetineq}\ee
Eq.(\ref{eq:targetineq}) relates the distance between two stationary states $\rho_1$ and $\rho_2$ to the distance between  two  quantum channels $\cT_1$ and $\cT_2$ from which those states arise as fixed points $\rho_i=\cT_i(\rho_i)$. A little thought reveals that such an inequality cannot hold in general if $\kappa$ is a constant merely depending on the chosen norm and possibly the dimension of the underlying space: let the $\cT_i$'s for instance be random dissipative perturbations of a unitary evolution. Then, irrespective of the size of the perturbation, there will generically not be any relation between the corresponding stationary states. So their distance can not be bounded in terms of the perturbation of the transition map, unless $\kappa$  depends on additional properties of at least one of the channels, e.g., $\kappa=\kappa(\cT_1)$. A property which suggests itself in this context is the rate of convergence: intuitively, if the Markov chain generated by $\cT_1$ converges rapidly towards $\rho_1$, one expects that the fixed point is rather robust with respect to perturbations of the transition map $\cT_1$. Conversely, if the mixing time is very long, i.e., if there are other states which are almost stationary already and  converge to $\rho_1$ only on a very large time scale, one expects a small perturbation to be sufficient in order to change the stationary state significantly.

We will follow two approaches which make this intuition rigorous. In Sec.\ref{FIX} we will directly derive an inequality of the form in Eq.(\ref{eq:targetineq}) where $\kappa$ is expressed in terms of a particular condition number which we will relate to spectral properties of the transition map in Sec.\ref{Sec:specbounds}.  Alternatively, we will in Sec.\ref{FIN} derive perturbation bounds for finite times for discrete as well as for continuous time quantum Markov processes with unique stationary state. Those bounds will be expressed in terms of an assumed exponential convergence bound. Hence, they are applicable whenever such a convergence bound can be obtained via one of the various existing tools such as logarithmic Sobolev inequalities \cite{LogSobolev}, $\chi^2$-divergence \cite{Chi2}, Hilbert's projective  metric \cite{Hilbert} or spectral theory \cite{szrewo}. If by any of those tools a time-scale is identified on which convergence is guaranteed, then the results in Sec.\ref{FIN} essentially provide a linear bound on the sensitivity coefficient in terms of that  mixing time bound.

\section{Preliminaries}\label{PREL}
\subsection{Quantum states and quantum evolutions}
\label{prel:not}
We begin with fixing the notation and terminology.
We will throughout consider finite-dimensional Hilbert spaces isomorphic to  $\C^d$ for some $d\in\mathbb{N}$. The notion of a \emph{state}  refers to a density matrix, i.e., a positive semidefinite matrix $\rho\in\cM_d(\C)$, $\rho\geq 0$ with unit trace $\tr[\rho]=1$. Here, $\cM_d(\C)$ denotes the space of complex valued $d\times d$ matrices. The objects of interest in this work are linear maps on $\cM_d(\C)$ for which we reserve the letters $\cT, \cE$ and $\cL$.  For each such map $\cT$ the   \emph{dual map} $\cT^*:\cM_d(\C)\ra \cM_d(\C)$ is defined  by imposing that $\forall A,B\in\cM_d(\C):\tr[\cT^*(A)B]=\tr[A\cT(B)]$. $\cT$ is called \emph{Hermiticity preserving} iff $\forall A\in\cM_d(\C): \cT(A)^\dagger=\cT(A^\dagger)$, \emph{positive} iff $A\geq 0\Rightarrow \cT(A)\geq 0$ and \emph{trace-preserving} iff $\forall A\in\cM_d(\C):\tr[\cT(A)]=\tr[A]$. The latter is equivalent to the fact that the dual map preserves the identity matrix $\1=\cT^*(\1)$. The identity map on $\cM_d(\C)$ will be denoted by $\id$. 

Our primary interest lies in \emph{quantum channels}, i.e., \emph{completely positive} and trace-preserving linear maps, which describe the time evolution of quantum systems for a single time step.  We will, however,  state all our results for maps which are positive but not necessary completely positive since the proofs do not require the stronger assumption of complete positivity.

Let now $\cT$ be any linear, trace-preserving and positive map on $\cM_d(\C)$. The spectrum ${\rm spec}[\cT]:=\{\lambda\in\C|\exists X:\cT(X)=\lambda X\}$ then contains $1$, is closed w.r.t.  complex conjugation  and is contained in the unit disc. That is, $\lambda\in {\rm spec}[\cT]\Rightarrow \bar{\lambda}\in{\rm spec}[\cT]$ and the spectral radius $\varrho(\cT):=\max\{|\lambda|\big|\lambda\in{\rm spec}[\cT]\}$ satisfies $\varrho(\cT)=1\in {\rm spec}[\cT]$. 

A state which satisfies $\rho=\cT(\rho)$ will be called a \emph{stationary state}. The set of stationary states is always non-empty and in fact spans the space of all fixed points of $\cT$. The projection onto this space will be denoted by $\cT^\infty$ and it can be expressed as a Cesàro mean via $$\cT^\infty=\lim_{n\ra\infty}\frac{1}{n}\sum_{k=1}^n \cT^k,$$
where $\cT^k=\cT\circ\ldots\circ\cT$ stands for the $k$-fold composition of $\cT$. Clearly, if $1$ is the only  eigenvalue of $\cT$ of modulus one, then this simplifies to $\cT^\infty=\lim_{n\ra\infty}\cT^n$. Note that the spectral properties of $\cT$ and $\cT^\infty$ guarantee that the map 
\be \cZ(\cT):=\big(\id-(\cT-\cT^\infty)\big)^{-1},\label{eq.fundmatdef}\ee
always exists.
For more details on spectral properties of (completely) positive maps we refer to \cite{WPG_eigenvalues}.

When applied to an initial state, the sequence $\{\cT^n\}_{n\in\mathbb{N}}$ can be regarded as a finite and homogeneous \emph{quantum Markov chain} with $\cT$ as its transition map. The classical case described by a stochastic matrix $S\in\cM_d(\mathbb{R}_+)$ can be embedded into this framework by fixing an orthonormal basis $\{|i\>\}_{i=1}^d$ and setting $\cT(\cdot)=\sum_{i,j=1}^d S_{i,j}\bracket{i}{\cdot}{i}\:\proj{j}{j}$.

\subsection{Norms and contraction coefficients}
For any $X\in\cM_d(\C)$ we denote by $\norm{X}{1}:=\tr\big[\sqrt{X^\dagger X}\big]$ the \emph{Schatten 1-norm} or \emph{trace norm} of $X$.  When applied to quantum states, the induced metric $(\rho_1,\rho_2)\mapsto\|\rho_1-\rho_2\|_1$ quantifies how well the two states can be distinguished in an optimally chosen  statistical experiment. 

For any linear map $\cL:\cM_d(\C)\ra\cM_d(\C)$ the induced 1-to-1-norm is defined as
$$\norm{\cL}{1\rightarrow1}:=\sup_{X\neq0}\frac{\norm{\cL(X)}{1}}{\norm{X}{1}},\  X\in\cM_d(\C).$$
By Gelfand's formula we can express the spectral radius of $\cL$ in terms of this norm as \cite{Bhatia}
\begin{align}
\varrho(\cL)=\lim_{n\rightarrow\infty}\Norm{\cL^n}{1\rightarrow1}^{1/n}.\label{specrad}
\end{align}
If $\cT$ is trace-preserving and positive, then $\|\cT\|_{1\ra 1}=1$.

We will frequently use the so called \emph{coefficient of ergodicity} or \emph{trace norm contraction coefficient} which is defined as 
$$\tau(\cL):=\sup_{\sigma^\dagger=\sigma\neq0\atop \trace{\sigma}=0} \frac{\Norm{\cL(\sigma)}{1}}{\Norm{\sigma}{1}}.$$
This quantity can equivalently be obtained via an optimization over orthogonal pure states \cite{Rus94},
\begin{align}
\tau(\cL)=\frac{1}{2}\sup_{\varphi\perp\psi}\Norm{\cL(\proj{\varphi}{\varphi})-\cL(\proj{\psi}{\psi})}{1}.\label{ergcoeff}
\end{align}
Here the supremum is taken over all pairs of orthogonal unit vectors. For linear maps which are Hermiticity preserving and trace-preserving  it follows readily from the definition of $\tau$ that
\begin{align}
\tau(\cL_1\circ \cL_2)\leq \tau(\cL_1)\tau(\cL_2),\  \mbox{and so }\  \tau(\cL^n)\leq\tau(\cL)^n\label{boundingpowers}
\end{align}
for all $n\in\mathbb{N}$.
Finally, note that  $0\leq\tau(\cT)\leq1$ if $\cT$ is positive and trace-preserving.
\section{Stability of Fixed Points} \label{FIX}

\subsection{The main inequality}
One of the first sensitivity analyses for fixed points of classical Markov chains was provided by Schweitzer \cite{Schweitzer} in terms of the so called \emph{fundamental matrix} of a classical Markov chain. Here, we generalize his approach to the quantum setting. The immediate analogue of Schweitzer's fundamental matrix is the map $\cZ(\cT):\cM_d(\C)\ra\cM_d(\C)$ defined in Eq.(\ref{eq.fundmatdef}). This leads to the main inequality:

 \begin{theorem}\label{SCH} Let $\cT_1,\cT_2:\cM_d(\C)\rightarrow\cM_d(\C)$ be trace-preserving, positive linear maps. For every stationary state $\rho_2$ of $\cT_2$ the  stationary state $\rho_1:=\cT_1^\infty(\rho_2)$ of $\cT_1$ satisfies
\be\Norm{\rho_1-\rho_2}{1}\leq\kappa\Norm{\cT_1-\cT_2}{1\rightarrow 1}\quad \mbox{with}\quad \kappa=\tau\big(\cZ(\cT_1)\big).\ee
\end{theorem}
\begin{proof} For all such pairs $\rho_1,\rho_2$ it holds that
\bea \cZ(\cT_1)^{-1}(\rho_1-\rho_2) &=& \big(\id-(\cT_1-\cT_1^\infty)\big)(\rho_1-\rho_2)\\
&=&\cT_1(\rho_2)-\rho_2 = (\cT_1-\cT_2)(\rho_2),
\eea
which leads to the identity
\be (\rho_1-\rho_2)=\cZ(\cT_1)\circ(\cT_1-\cT_2)(\rho_2).\label{eq:mainidentity}\ee
Taking the Schatten 1-norm on both sides and abbreviating $\sigma:=(\cT_1-\cT_2)(\rho_2)$ we can write
\be \|\rho_1-\rho_2\|_1=\frac{\|\cZ(\cT_1)(\sigma)\|_1}{\|\sigma\|_1}\ \frac{\|(\cT_1-\cT_2)(\rho_2)\|_1}{\|\rho_2\|_1}, \ee
from which we obtain the claimed inequality by taking the supremum over all $\rho_2\in\cM_d(\C)$ and over all traceless Hermitian $\sigma$.\qed
\end{proof}

Evidently, the identity in Eq.(\ref{eq:mainidentity}) can be used to derive a plethora of different norm bounds (cf. \cite{Overview} for an overview on different approaches for classical Markov chains). Here we focus on the trace norm since this seems to be the most natural choice in the quantum context. In addition the trace norm dominates all other unitarily invariant norms on $\cM_d(\C)$ \cite{Bhatia} and makes the obtained bounds in this sense the strongest possible ones.

In the following proposition we bound the condition number  of Thm.\ref{SCH} in terms of a better studied object \cite{Rus94,Chi2,Hilbert} with an operational meaning, namely the trace-norm contraction coefficient of the quantum channel:

\begin{proposition}Let $\cT$ be a trace-preserving, positive linear map on $\cM_d(\C)$ with a unique stationary state. Then
\be\tau\big(\cZ(\cT)\big)\leq \big(1-\tau(\cT)\big)^{-1}.\ee
\end{proposition}

\begin{proof} We express $\cZ(\cT)$ via its von Neumann series expansion $\cZ=\sum_{k=0}^{\infty}\left(\cT-\cT^\infty\right)^k$ and get that
\begin{align}
\tau(\cZ)&=\sup_{\sigma^\dagger=\sigma\atop\trace{\sigma}=0}\frac{\Norm{\sum_{k=0}^\infty\left(\cT-\cT^\infty\right)^k(\sigma)}{1}}{\Norm{\sigma}{1}}\nonumber\\
&\leq\sum_{k=0}^{\infty}\sup_{\sigma^\dagger=\sigma\atop\trace{\sigma}=0}\frac{\Norm{\left(\cT-\cT^\infty\right)^k(\sigma)}{1}}{\Norm{\sigma}{1}}=\sum_{k=0}^{\infty}\tau(\cT^k)\label{trless}\\
&\leq\sum_{k=0}^{\infty}\left[\tau(\cT)\right]^k=\frac{1}{1-\tau(\cT)},\quad\mbox{if}\quad \tau(\cT)<1.\nonumber
\end{align}
To obtain Eq.~\eqref{trless} we used that $(\cT-\cT^\infty)^k=\cT^k-\cT^\infty$ if $k>0$ and that $\trace{\sigma}=0$ implies $\cT^\infty(\sigma)=0$ since uniqueness of the fixed point means that $\cT^\infty$ acts as $X\mapsto\tr{[X]}\rho$. \qed
\end{proof}

Note that $\tau(\cT)$ has an operational meaning. Since $\tau(\cT)=\frac{1}{2}\sup_{\varphi\perp\psi}\Norm{\cT(\proj{\varphi}{\varphi})-\cT(\proj{\psi}{\psi})}{1}$ by Eq.~\eqref{ergcoeff}, it is directly related to the maximum probability with which two orthogonal inputs can be distinguished at the output of $\cT$.
\subsection{Spectral bounds on $\tau(\cZ)$}\label{Sec:specbounds}
In this subsection we prove that the sensitivity of the set of stationary states of a quantum Markov chain to perturbations is related to the closeness of the subdominant eigenvalues to $1$. More precisely, we show that if there exists a subdominant eigenvalue of $\cT$ close to 1, then the chain is ill conditioned in the sense that $\tau(\cZ)$ is large. On the other hand if all eigenvalues are well separated from $1$, the process is well conditioned. The following theorem quantifies this observation.
We note that the relevant spectral quantity is not equal to the \emph{spectral gap} $\min\{1-|\lambda|\big|\lambda\in{\rm spec}[\cT]\setminus\{1\}\}$ which also appears frequently in convergence analyses.
\begin{theorem}
Let $\cT$ be a trace-preserving, positive linear map on $\cM_d(\C)$ and $\Lambda:={\rm spec}[\cT]\backslash\{1\}$ the set of its non-unit eigenvalues. Then 
\be\frac{1}{\min_{\lambda\in\Lambda}\abs{1-\lambda}}\leq\tau\big(\cZ(\cT)\big)\leq\frac{2(5\pi/3+2\sqrt{2})d^3}{\min_{\lambda\in\Lambda}\abs{1-\lambda}}.\ee
\end{theorem}
\begin{proof}
We begin with proving the left hand inequality---guided by the techniques developed for classical Markov chains in \cite{Sen2}. We abbreviate $\cZ:=\cZ(\cT)$ and note  that $\cZ$ is trace-preserving, since $\cZ^{-1}$ is trace-preserving and therefore $\tr[\cZ(X)]=\tr[\cZ^{-1}\circ\cZ(X)]=\tr[X]$. Consequently, $\cZ^*(\1)=\1$.  We write $\cP$ for the projection onto the invariant subspace of $\cM_d(\C)$ corresponding to the eigenvalue $1$ of $\cZ$. Note that this implies that
\begin{align*}
(\cZ- \cP)^k=\cZ^k\circ(\id- \cP).
\end{align*}
Using the fact that any matrix $\sigma$ can be expressed as a sum of a Hermitian matrix $\sigma_+:=(\sigma+\sigma^\dagger)/2$ and a skew-Hermitian matrix $i\sigma_-:=(\sigma-\sigma^\dagger)/2$, i.e., $\sigma=\sigma_++i\sigma_-$,  we can bound 
\begin{align}
&\Norm{(\cZ- \cP)^k}{1\rightarrow 1}=\Norm{\cZ^k\circ(\id- \cP)}{1\rightarrow 1}\nonumber\\
&=\sup_{\sigma=\sigma_++i\sigma_-}\frac{\Norm{\cZ^k\circ(\id- \cP)(\sigma_++i\sigma_-)}{1}}{\Norm{\sigma_++i\sigma_-}{1}}\nonumber\\
&\leq\sup_{\sigma_+,\sigma_-}\frac{\Norm{\cZ^k\circ(\id- \cP)(\sigma_+)}{1}}{\Norm{\sigma_+}{1}}+\frac{\Norm{\cZ^k\circ(\id- \cP)(\sigma_-)}{1}}{\Norm{\sigma_-}{1}}\label{added}\\
&\leq 2\sup_{\sigma^\dagger=\sigma\atop \trace{\sigma}=0} \frac{\Norm{\cZ^k(\sigma)}{1}}{\Norm{\sigma}{1}} \Norm{\id-\cP}{1\rightarrow 1}\label{usehermprop}\\
&=2\tau(\cZ^k) \Norm{\id-\cP}{1\rightarrow 1}\nonumber\\
&\leq 2\tau(\cZ)^k \Norm{\id-\cP}{1\rightarrow 1}\label{usepowers}.
\end{align}

To obtain Eq.~\eqref{added} we apply the triangle inequality in the numerator and note that again by the triangle inequality $\Norm{\sigma_i}{1}\leq\Norm{\sigma}{1}$, $i\in\{+,-\}$ holds to bound the denominator.
Inequality~\eqref{usehermprop} exhibits the fact that both $\sigma_+$ and $\sigma_-$ are Hermitian and that $(\id-\cP)(\sigma)$ is traceless for all $\sigma$. To obtain Eq.~\eqref{usepowers} we used Eq.~\eqref{boundingpowers}. Taking the $k$ 'th root and the limit  $k\rightarrow\infty$ on both sides of the above derivation, we conclude with Eq.~\eqref{specrad} that
\begin{align}
\varrho(\cZ-\cP)\leq\tau(\cZ)\label{spec}.
\end{align}
That is, $\tau(\cZ)$ provides an upper bound on the modulus of all non-unit eigenvalues of $\cZ$.
Finally, note that the spectrum of $\cZ$ is given by ${\rm spec}[\cZ]=\{1\}\cup\{\frac{1}{1-\lambda}\}_{\lambda\in\Lambda}$ from which the lower bound in the theorem follows.

For the upper bound we use known results from non-classical spectral theory. The core observation is that the map $\Delta:=\cT-\cT^\infty$ is \emph{power bounded} since 
\bea\nonumber\Norm{\Delta^n}{1\rightarrow1}&=&\Norm{\cT^n-\cT^\infty}{1\rightarrow1}\\
&\leq&\Norm{\cT^n}{1\rightarrow1}+\Norm{\cT^\infty}{1\rightarrow1}=2.\nonumber\eea


In \cite{Rachid} it has been shown that the resolvent of a general power bounded operator $\Delta$, which acts on a complex $D$-dimensional Banach space and whose spectrum is contained in the open unit disc, satisfies 
\begin{align}
\Norm{(\mu\:\id-\Delta)^{-1}}{}\leq \frac{C\left(\frac{5\pi}{3}+2\sqrt{2}\right) D^{3/2}}{\min_{\lambda\in{\rm spec}[\Delta]}|\mu-\lambda|}\label{nonclass},
\end{align}
for all $\abs{\mu}\geq1$ and $C:=\sup_n\Norm{\Delta^n}{}$, where $\Norm{\cdot}{}$ denotes the usual operator norm induced by the norm of the Banach space. The core observation in \cite{Rachid} is that one can bound the norm $\Norm{(\mu\:\id-\Delta)^{-1}}{}$ by employing a Wiener algebra functional calculus and bounding $\norm{\frac{1}{\mu-z}}{W/{mW}}:=\inf\{\norm{\frac{1}{\mu-z}+mg}{W}|\:g\in W\}$, where $\Norm{\cdot}{W}$ denotes the  Wiener norm and $m\neq0$ is the minimal degree polynomial annihilating $\Delta$, i.e., $m(\Delta)=0$. For more details concerning the techniques employed see \cite{Rachid,AC,szrewo} and references therein.

Suppose for now, that the only eigenvalue of $\cT$ of magnitude one is $1$. Then, all eigenvalues of $\Delta$ are contained in the open unit disc. Setting $D=d^2$, $\mu=1$, observing that ${\rm spec}[\Delta]=\Lambda$ and bounding $C\leq 2$, Eq.~\eqref{nonclass} specializes to the upper bound claimed in the theorem.

To incorporate the case when $\cT$ has eigenvalues of magnitude one other than $1$, i.e., when the spectrum of $\Delta$ is merely contained in the closed unit disc, we employ an argument based on continuity. We consider a map $\cT_\epsilon$ whose spectrum differs from the one of $\cT$ in that the peripheral eigenvalues other than $1$ of $\cT$ are shifted ``by $\epsilon$\rq{}\rq{} radially towards the center of the unit disc. More precisely, we define $\cT_\epsilon:=\cT-\epsilon(\cT_\phi-\cT^\infty)$, where $\cT_\phi$ denotes the part of the spectral decomposition of $\cT$ which belongs to all eigenvalues of magnitude one, i.e., if $\cT=\sum_k\lambda_k \cP_k$ then $\cT_\phi=\sum_{k:\abs{\lambda_k}=1}\lambda_k\cP_k$. Exploiting the relations between $\cT,\cT^\infty$ and $\cT_\phi$ we can show that 
\bea \Delta_\epsilon^n &:=& \big(\cT_\epsilon-\cT^\infty\big)^n\nonumber\\
&=&\cT^n-(1-\epsilon)^n\cT^\infty+[(1-\epsilon)^n-1]\cT_\phi^n.\nonumber\eea 
Since the involved maps are all positive and trace-preserving, and thus have unit norm, this implies that $\Delta_\epsilon$ is power bounded with $\Norm{\Delta_\epsilon^n}{1\ra 1}\leq2$ as before. Thus, for $\mu=1$ and any (small enough) $\epsilon>0$ the above assertion~\eqref{nonclass} holds for $\Delta_\epsilon$. Then by continuity  the statement stays true even for $\epsilon=0$. 

\qed
\end{proof}
\section{Finite time perturbation bounds}\label{FIN}
So far, we have analysed the stability of the fixed point of a quantum channel and in this sense the robustness of the asymptotic time evolution of the corresponding quantum Markov chain. In this section we will extend the analysis to finite times, first for discrete and then for continuous time evolutions. A second point in which the following approach differs from the previous one is that it uses the assumption of an exponential convergence bound as an additional ingredient.
\subsection{Evolution in discrete time}
\begin{theorem}\label{MIT} For $n\in\mathbb{N}_0$ let $\rho_n:=\cT^n(\rho_0)$ and  $\sigma_n:=\cE^n(\sigma_0)$ be the evolution of two density matrices with respect to two positive and trace-preserving linear maps $\cT,\cE:\cM_d(\C)\ra\cM_d(\C)$.
If $\cT$ has a unique stationary state and $\Norm{\cT^n-\cT^\infty}{1\rightarrow 1}\leq K\cdot \mu^n$ for $K\geq 0$, $\mu<1$ and all $n\in\mathbb{N}_0$, then we can bound the distance between the evolved states by
\begin{align*}
&\Norm{\rho_n-\sigma_n}{1}\leq\\
&\begin{cases}
\Norm{\rho_0-\sigma_0}{1}+n \Norm{\cE-\cT}{1\rightarrow 1}\ \textnormal{for}\  n\leq\hat{n}\\
 K\mu^n\Norm{\rho_0-\sigma_0}{1}+\left(\hat{n}+K \frac{\mu^{\hat{n}}-\mu^n}{1-\mu}\right)\Norm{\cE-\cT}{1\rightarrow 1}\:\textnormal{for}\  n>\hat{n},
\end{cases}
\end{align*}
where $\hat{n}:=\left\lceil\frac{\log(1/K)}{\log(\mu)}\right\rceil$.
\end{theorem}

\begin{remark}
Before we prove this statement, let us mention known pairs  $(K,\mu)$ to which the theorem might be applied. For definitions and further details we refer to the references:

(i) For $K=1$ one can choose $\mu=\tanh(\Delta/4)$, where $\Delta$ is the \emph{projective diameter} of the map $\cT$, measured in terms of Hilbert's projective metric \cite{Hilbert}.

(ii) For $K=\sup_\rho \big[\chi_k^2(\rho,\sigma)\big]^{1/2}$ a particular $\chi^2$\emph{-divergence} and $\sigma$ the stationary state of $\cT$, we can choose $\mu$ to be the second largest singular value of $\Omega:=[\Omega_\sigma^k]^{1/2}\circ\cT\circ[\Omega_\sigma^k]^{-1/2}$ where $\Omega_\sigma^k$ is a map on which the chosen $\chi^2$-divergence is based on \cite{Chi2}. 

If $\lambda_{min}$ is the smallest eigenvalue of $\sigma$, a particular choice results in $K=(\lambda_{min}^{-1}-1)^{1/2}$ and $\Omega(X)=\sigma^{-1/4}\cT\big(\sigma^{1/4}X\sigma^{1/4}\big)\sigma^{-1/4}$.

(iii) If $\lambda_{min}$ is the smallest eigenvalue of the stationary state of $\cT$, we can choose $K=\sqrt{-2\log\lambda_{min}}$ and $\mu$ determined by a logarithmic Sobolev inequality \cite{LogSobolev}. Strictly speaking, those bounds apply to the continuous time case, which we discuss in Thm.\ref{MIT2} below.

(iv) If there is a similarity transformation such that $S\circ\cT\circ S^{-1}$ is a normal operator on $\cM_d(\C)$, we can choose $\mu:=\max\{|\lambda|\big|\lambda\in{\rm spec}[\cT]\setminus\{1\}\}$ and $K=\sqrt{2d}\kappa_\cT$, where $\kappa_\cT:=\|S\otimes S^{-1}\|_{2\ra 2}$. The latter can be upper bounded by $\kappa_\cT\leq\lambda_{min}^{-1/2}$ if $\cT$ satisfies \emph{detailed balance} w.r.t. to its stationary state.

(v) Finally, we note that the assumption that $\Norm{\cT^n-\cT^\infty}{1\rightarrow 1}\leq K\cdot \mu^n$ for $K\geq 0$, $\mu<1$ of Thm.4 implies that all non-unit eigenvalues of $\cT$ are contained in the open unit disc. In this situation more elaborate bounds, which only depend on the spectrum of $\cT$ can be given. In \cite{szrewo} a Wiener algebra functional calculus is employed to obtain spectral convergence bounds for classical and quantum Markov chains. The techniques of \cite{szrewo} are new even to the theory of classical Markov chains and do not rely on additional assumptions such as \emph{detailed balance}. The derivation of Corollary IV.4 \cite{szrewo} yields that for any $\mu$ such that $\abs{\lambda_i}<\mu<1\  \forall i$, where ${(\lambda_i)}_{i=1,...,d^2}$ denote the eigenvalues of $\cT-\cT^\infty$ it holds that
\begin{align}
\Norm{\cT^n-\cT^\infty}{1\rightarrow1}
&\leq\frac{4 e\sqrt{\abs{m}}}{\left(1-\mu\right)^{3/2}}\:\sup_{\abs{z}=\mu}\abs{\prod_{i\in m}\frac{1-\bar{\lambda}_iz}{z-\lambda_i}}\:\mu^{n+1}\nonumber\\
&\leq\frac{4 e\sqrt{\abs{m}}}{\left(1-\mu\right)^{3/2}}\:\prod_{i\in m}\frac{1-\mu\abs{\lambda_i}}{\mu-\abs{\lambda_i}}\:\mu^{n+1}.\label{eq:superbound}
\end{align}
Here, $m$ denotes the minimal polynomial of $\cT-\cT^\infty$ and $\abs{m}$ is the number of linear factors in $m$. The product in Eq.(\ref{eq:superbound}) is taken over all $i$ such that the corresponding factor $(z-\lambda_i)$ occurs in the prime factorization of $m$.

Thm.\ref{MIT} together with Eq.\eqref{eq:superbound} provide a purely spectral bound on the sensitivity of a Markov chain under perturbations. Even  compared to the results for classical Markov chains in \cite{Mit1} (on which our derivation of Thm.\ref{MIT} builds), bounds based on $\eqref{eq:superbound}$ yield a significant improvement (compare \cite{Mit1}, Thm.4.1). Our bound proves that the distance of the subdominant eigenvalues of $\cT$ to the spectral radius of $\cT-\cT^\infty$ determines the sensitivity of the chain to perturbation, while their mutual distances, i.e., the quantities $\abs{\lambda_i-\lambda_j}$ for general $i,j$ are not relevant (compare \cite{Mit1}, Thm.4.1). 
We refer to \cite{szrewo} for a discussion of Eq.\eqref{eq:superbound} and related results.

It is also possible to use Corollary IV.4 of \cite{szrewo} directly to derive stability estimates. Note, however, that if in $\Norm{\cT^n-\cT^\infty}{1\rightarrow 1}\leq K\cdot \mu^n$ we allow that $\mu$ equals the spectral radius of $\cT-\cT^\infty$ then the prefactor will depend on $n$. More precisely, \cite{szrewo} Thm. III.2 yields that in this case $K=K(n)=Cn^{d_\mu-1}$, where $C$ does not depend on $n$ and $d_\mu$ denotes the size of the largest Jordan block of any eigenvalue of magnitude $\mu$. It is not difficult to extend the derivation of Thm.~\ref{MIT} to the situation, where $\mu$ is the spectral radius of $\cT-\cT^\infty$ and $K(n)=Cn^{d_\mu-1}$.

\end{remark}

\begin{proof} of Thm.\ref{MIT}. 
The proof is guided by techniques used in \cite{Mit1} for classical Markov processes.
 First we note that in general for linear maps $\cT,\cE,$ 
\begin{align*}
\cE^n=\cT^n+\sum_{i=0}^{n-1}\cT^{n-i-1}\circ(\cE-\cT)\circ\cE^{i}\quad n\geq1
\end{align*}
holds, which can easily be shown by induction. Applying the above to the state $\sigma_0$ and subtracting $\rho_n$ from both sides gives
\begin{align*}
\sigma_n-\rho_n=\cT^n(\sigma_0-\rho_0)+\sum_{i=0}^{n-1}\cT^{n-i-1}\circ(\cE-\cT)(\sigma_i)
\end{align*}
from which we conclude that
\begin{align}
&\Norm{\sigma_n-\rho_n}{1}\leq\label{norm}\\
&\qquad\Norm{\cT^n(\sigma_0-\rho_0)}{1}+\sum_{i=0}^{n-1}\Norm{\cT^{n-i-1}\circ(\cE-\cT)(\sigma_i)}{1}.\nonumber
\end{align}
We now find upper bounds for the norm terms appearing on the right hand side of  Eq. \eqref{norm}.
The fact that $\cE(\sigma_i)-\cT(\sigma_i)$ is Hermitian and traceless implies that
\bea\nonumber
\Norm{\cT^{n-i-1}\circ(\cE-\cT)(\sigma^i)}{1}&\leq&\tau(\cT^{n-i-1})\Norm{\cE-\cT}{1\rightarrow1},\\ \nonumber
\mbox{and }\ 
\Norm{\cT^n(\sigma_0-\rho_0)}{1}&\leq&\tau(\cT^n)\Norm{\rho_0-\sigma_0}{1}.
\eea
Thus, from Eq.~\eqref{norm} we conclude that
\begin{align}
\Norm{\rho_n-\sigma_n}{1}\leq\tau(\cT^n)\Norm{\rho_0-\sigma_0}{1}+\Norm{\cE-\cT}{1\rightarrow1}\sum_{i=0}^{n-1}\tau(\cT^i).\label{finish}
\end{align}
The term $\tau(\cT^n)$  can in turn be bounded using Eq.~\eqref{ergcoeff}  and the assumed convergence properties of $\cT$ by
\begin{align*}
\tau(\cT^n)&=\sup_{\sigma^\dagger=\sigma\atop \trace{\sigma}=0} \frac{\Norm{\cT^n(\sigma)}{1}}{\Norm{\sigma}{1}}=\frac{1}{2}\sup_{\ket{\phi},\ket{\psi}}\Norm{\cT^n{(\phi)}-\cT^n{(\psi)}}{1}\\
&\leq\sup_{\ket{\phi}}\Norm{\cT^n{(\phi)}-\cT^\infty(\phi)}{1}\leq K\cdot\mu^n.
\end{align*}
Note that the first inequality requires uniqueness of the stationary state, i.e., that $\cT^\infty(\phi)=\cT^\infty(\psi)$.

Alternatively, we can use that $\cT^n$ is trace-preserving and positive, so that in total
\begin{align*}
\tau{(\cT^n)}\leq
\begin{cases}
1\ &\textnormal{for}\ n<{\hat{n}}\\
K\cdot\mu^n\ &\textnormal{for}\ n\geq\hat{n}.
\end{cases}
\end{align*}
We now find a suitable upper bound on $\sum_i\tau(\cT^i)$ by always choosing the better of the two bounds for $\tau(\cT^i)$. In this way we obtain
\begin{align}
\sum_{i=0}^{n-1}\tau(\cT^i)&\leq\hat{n}+\sum_{i=\hat{n}}^{n-1}\tau(\cT^i)\leq\hat{n}+K\cdot\mu^{\hat{n}}\sum_{i=0}^{n-\hat{n}-1}\mu^i\nonumber\\
&=\hat{n}+K\cdot\mu^{\hat{n}}\frac{1-\mu^{n-\hat{n}}}{1-\mu}.\label{altapp}
\end{align}
Plugging this expression into Eq.~\eqref{finish} and again choosing the better bound for $\tau(\cT^n)$ concludes the proof of the theorem.\qed
\end{proof}
If we take the limit $n\rightarrow\infty$ in Thm.~\ref{MIT} and use that $K\cdot\mu^{\hat{n}}\leq1$ is by definition of $\hat{n}$ basically an equality, we obtain a perturbation bound for the asymptotic states:
\begin{corollary} \label{MITC} Under the conditions of Thm.~\ref{MIT} 
\be\label{eq:Cor:normdist}\limsup_{n\ra\infty}\Norm{\rho_n-\sigma_n}{1}\leq\left(\hat{n}+\frac{1}{1-\mu}\right)\Norm{\cE-\cT}{1\rightarrow 1}.\ee
\end{corollary}

\subsection{Evolution in continuous time}
The following is the quantum counterpart of  the results on classical Markov chains in \cite{Mit2}:  
\begin{theorem} \label{MIT2}
Let $\cT^t=e^{t\mathfrak{L}_\cT}$ and $\cE^t=e^{t\mathfrak{L}_\cE}$ with $t\in\mathbbm{R}_+$ be two  one-parameter semi groups of positive and trace-preserving linear maps on $\cM_d(\C)$. Write $\rho(t):=\cT^t(\rho_0)$ and $\sigma(t):=\cE^t(\sigma_0)$ for the evolution of two density matrices and assume that $\cT^t$ has a unique stationary state and that 
$\forall t>0:\Norm{\cT^t-\cT^\infty}{1\ra 1}\leq K e^{-\nu\ t}$ for some $K,\nu>0$. Then
\begin{align*}&\Norm{\rho(t)-\sigma(t)}{1}\leq\\
&\quad\begin{cases}
 \Norm{\rho_0-\sigma_0}{1}+t \Norm{\mathfrak{L}_\cE-\mathfrak{L}_\cT}{1\rightarrow 1},
 \quad\textnormal{for}\ t<\hat{t}\\
 Ke^{-\nu t}\Norm{\rho_0-\sigma_0}{1}+\frac{\log(K)+1-Ke^{-\nu t}}{\nu}\Norm{\mathfrak{L}_\cE-\mathfrak{L}_\cT}{1\rightarrow 1}\\
\qquad\textnormal{for}\ t\geq\hat{t}
\end{cases},
\end{align*}
where $\hat{t}:=\frac{\log(K)}{\nu}$.
\end{theorem}
\begin{proof} The proof goes along the lines of the proof of Thm.~\ref{MIT}. The difference between two dynamical semi-groups can be expressed using their generators as \cite{Phillips} 
\begin{align*}
\cE_t=\cT_t+\int_0^t{\cT_{t-s}\circ\left(\mathfrak{L}_\cE-\mathfrak{L}_{\cT}\right)\circ\cE_s}\:\d{s}.
\end{align*}
Following the derivation of Thm.~\ref{MIT} and using that $\forall X:\tr[(\mathfrak{L}_\cE-\mathfrak{L}_{\cT})(X)]=0$, we obtain the continuous time analogue of Eq.~\eqref{finish},
\begin{align*}
&\Norm{\rho(t)-\sigma(t)}{1}\leq\\
&\qquad\tau(\cT_t)\Norm{\rho_0-\sigma_0}{1}+\Norm{\mathfrak{L}_\cE-\mathfrak{L}_\cT}{1\ra 1}\int_0^t\tau(\cT_u)\:\d{u}.
\end{align*}
Again, it is possible to state upper bounds for $\tau(\cT_t)$ for small and large $t$, respectively. We have that
\begin{align*}
\tau{(\cT_t)}\leq
\begin{cases}
1\ &\textnormal{for}\ t\leq{\hat{t}}\\
K\cdot e^{-\nu t}\ &\textnormal{for}\ t>\hat{t},
\end{cases}
\end{align*}
where $\hat{t}:=\frac{\log(K)}{\nu}.$
The proof is then concluded following exactly the same lines as in  the proof of Thm.~\ref{MIT}.\qed
\end{proof}
Again we can consider the limit $t\ra\infty$ and thereby obtain a perturbation bound for the asymptotic evolution in terms of the distance between the generators and as a function of the convergence rate $\nu$:
\begin{corollary} Under the conditions of Thm.~\ref{MIT2} 
\begin{align*}
\limsup_{t\ra\infty}\Norm{\rho(t)-\sigma(t)}{1}\leq\frac{\log(K)+1}{\nu}\Norm{\mathfrak{L}_\cE-\mathfrak{L}_\cT}{1\rightarrow 1}.
\end{align*}
\end{corollary}

\section{Outlook}\label{CON}

We have established general perturbation bounds for fixed points of quantum Markov chains. The results focus on the trace norm, but it is clear from their derivation, that analogous bounds can be obtained for essentially any norm. For practical purposes and large systems, the derived bounds may be weaker than desired---owing to the fact that we do not impose and exploit any additional structure of transition map and perturbation. Investigating  bounds in more structured frameworks, where for instance Liouvillians as well as perturbations are geometrically local, seems to be a worthwhile direction for future studies.

We have also seen that perturbation bounds are linked to convergence bounds so that stronger perturbation bounds can be obtained from better convergence bounds. A detailed analysis of the latter, leading to bounds of the form in Eq.(\ref{eq:superbound}), will be given \cite{szrewo}. 

Clearly, one may also exploit the relation in the other direction and use the derived perturbation bounds in order to obtain lower bounds on mixing times for quantum Markov processes.

\ \\

\subparagraph{Acknowledgements:} 
We acknowledge financial support from the European project QUEVADIS, the CHIST-ERA/BMBF project CQC and the QCCC programme of the Elite Network of Bavaria.

\bibliographystyle{abbrv}

\end{document}